\documentclass[letterpaper, 10 pt, conference]{ieeeconf}  
\IEEEoverridecommandlockouts                           
\usepackage[usenames,dvipsnames]{xcolor}
\usepackage{amsfonts}
\usepackage[sort]{cite}
\usepackage{amsmath} 
\usepackage{amssymb}  
\usepackage{graphicx}
\usepackage{epsfig}
\usepackage{epstopdf}
\usepackage[normal]{subfigure}
\usepackage{mathtools}
\usepackage{etoolbox}
\usepackage{lipsum}
\usepackage{caption}
\usepackage{comment}
\usepackage{soul}
\usepackage[colorlinks=true]{hyperref}
\usepackage[normal]{subfigure}
\usepackage{verbatim}

\usepackage{algorithm,algorithmic}

\def\cov{\mathop{\mathrm{cov}}}  
\def\diag{\mathop{\mathrm{diag}}}  
\def\trace{\mathop{\mathrm{tr}}}  

\newcommand{\Cov}{\mathrm{cov}}

\let\bbordermatrix\bordermatrix
\patchcmd{\bbordermatrix}{8.75}{4.75}{}{}
\patchcmd{\bbordermatrix}{\left(}{\left[}{}{}
\patchcmd{\bbordermatrix}{\right)}{\right]}{}{}
\renewcommand{\thefigure}{\thesection.\arabic{figure}}
\renewcommand{\theequation}{\thesection.\arabic{equation}}
\numberwithin{equation}{section}

\hypersetup{
     colorlinks = true,
     linkcolor = [rgb]{0,0,1},
     anchorcolor = [rgb]{0,0,1},
     citecolor = [rgb]{0.9,0.5,0},
     filecolor = [rgb]{0,0,1},
     pagecolor = [rgb]{1,1,0},
     urlcolor = [rgb]{0.9,0.5,0},
     bookmarks,
        bookmarksopen = true,
        bookmarksnumbered = true,
        breaklinks = true,
        linktocpage,
        pagebackref,
        colorlinks = true,
        linkcolor = [rgb]{0.2,0.6,0.2},
        urlcolor  = [rgb]{0,0,1},
        citecolor = [rgb]{0.9,0.5,0},
        anchorcolor = [rgb]{0.2,0.6,0.2},
        hyperindex = true,
        hyperfigures
}

\newcommand{\T}{^{\mbox{\tiny T}}}

\newtheorem{theorem}{\bf Theorem}
\newtheorem{lemma}{\bf Lemma}
\newtheorem{definition}{\bf Definition}

\newtheorem{remark}{\bf Remark}
\newtheorem{problem}{\bf Problem}

\newtheorem{example}{\bf Example}

\newcommand{\sr}{\stackrel}

\newcommand{\rar}{\rightarrow}

\newcommand{\tri}{\sr{\bigtriangleup}{=}}

\newcommand{\noi}{\noindent}

\newcommand{\be}{\begin{equation}}
\newcommand{\ee}{\end{equation}}
\newcommand{\bea}{\begin{eqnarray}}
\newcommand{\eea}{\end{eqnarray}}
\newcommand{\bes}{\begin{eqnarray*}}
\newcommand{\ees}{\end{eqnarray*}}
\newcommand{\ba}{\begin{align}}
\newcommand{\ea}{\end{align}}

\newcommand{\bfi}{\begin{figure}}
\newcommand{\bfit}{\begin{figure}[t]}
\newcommand{\bfib}{\begin{figure}[b]}
\newcommand{\bfih}{\begin{figure}[h]}
\newcommand{\bfip}{\begin{figure}[p]}
\newcommand{\efi}{\end{figure}}

\newcommand{\bi}{\begin{itemize}}
\newcommand{\ei}{\end{itemize}}
\newcommand{\ben}{\begin{enumerate}}
\newcommand{\een}{\end{enumerate}}

\newcommand{\bp}{\begin{problem}}
\newcommand{\ep}{\end{problem}}
\newcommand{\hso}{\hspace{.1in}}
\newcommand{\hst}{\hspace{.2in}}

\newenvironment{list4}{
  \begin{list}{$\bullet$}{
      \setlength{\itemsep}{0.05cm}
      \setlength{\labelsep}{0.2cm}
      \setlength{\labelwidth}{0.3cm}
      \setlength{\parsep}{0in}
      \setlength{\parskip}{0in}
      \setlength{\topsep}{0in}
      \setlength{\partopsep}{0in}
      \setlength{\leftmargin}{0.17in}}}
      {\end{list}}
      
 \newenvironment{list5}{
  \begin{list}{$\bullet$}{
      \setlength{\itemsep}{0.05cm}
      \setlength{\labelsep}{0.2cm}
      \setlength{\labelwidth}{0.6cm}
      \setlength{\parsep}{0in}
      \setlength{\parskip}{0.0in}
      \setlength{\topsep}{0in}
      \setlength{\partopsep}{0in}
      \setlength{\leftmargin}{0.3in}}}
      {\end{list}}

\begin{document}

\title{Joint Nonanticipative Rate Distortion Function for a Tuple of Random Processes with Individual Fidelity Criteria}

\author{Charalambos D. Charalambous and Evagoras Stylianou   \\ 
\thanks{C. D. Charalambous is with the Department of Electrical and Computer Engineering, University of Cyprus, Nicosia, Cyprus. Emails: {\tt  chadcha@ucy.ac.cy}.}
\thanks{Evagoras Stylianou is  with the Department of Electrical and Computer Engineering, Technical University of Munich.      Email: {\tt evagoras.stylianou@tum.de}.}
\thanks {This work was supported in parts by  the European Regional Development Fund and the Republic of Cyprus through the Research Promotion Foundation Projects EXCELLENCE/1216/0365 and  EXCELLENCE/1216/0296.}
}

\maketitle
%
%
%
%
\begin{abstract} 
The  joint nonanticipative rate distortion function (NRDF) for a tuple of random processes with individual fidelity criteria is considered. Structural properties of optimal test channel distributions are derived. Further, for the application example of  the joint NRDF of a  tuple of jointly multivariate Gaussian Markov {processes} with individual  square-error fidelity criteria,  a   realization of the reproduction processes  which induces the optimal test channel distribution is derived,  and the corresponding  joint NRDF is characterized. The analysis of the simplest example,  of a  tuple of  scalar correlated Markov  processes,  illustrates many of the challenging aspects of such problems. 
\end{abstract}

%
%
%
%
\section{Introduction}\label{sec:intro}
Presently of much  interest  in information theory and in the theory of information transmission for control systems applications,  is the  Gorbunov and Pinsker  \cite{gorbunov-pinsker1973}  nonanticipatory epsilon entropy and message generation rates of a discrete-time random process $X^n \tri \{X_1, X_2, \ldots, X_n\}$,  $X_t(\omega),  t =1, 2, \ldots, n,  \omega \in \Omega$ taking values in ${\mathbb X}$,  with     joint probability distribution ${\bf P}_{X^n}$,   subject to   a fidelity criterion,   
\begin{align}
\frac{1}{n} {\bf E} \Big\{d_{n}(X^n, Y^n)  \Big\}\leq \epsilon,   \hso d_n(x^n,y^n) \in [0,\infty)   \label{dis_1}
\end{align}
of  reconstructing $X^n $ by another random process  $Y^n \tri \{Y_1, Y_2, \ldots, Y_n\}$,  $Y_t(\omega),  t =1, 2, \ldots, n,  \omega \in \Omega$ taking values in ${\mathbb Y}\subseteq {\mathbb X}$. Nonanticipatory entropy is often  described under the designated name nonanticipative or sequential rate distortion function (RDF) \cite{tatikonda:2004,derpich:2012,charalambous-stavrou-ahmed2014ieeetac,tanaka:2017,stavrou-charalambous-charalambous-loyka2018siam}.   However,  the  name for this quantity does alter the situation,  that it corresponds to    a variant of Shannon's   \cite{shannon1959} ``rate of creating information with respect to a fidelity'',  often designated by the name rate distortion function (RDF)  \cite{gallager1968}. Shannon's RDF is the information theoretic definition of  the operational definition,  ``the   optimal performance theoretically attainable'' (OPTA)   (i.e.,  the infimum of rates of creating information) by noncausal codes subject  to a fidelity.  

Gorbunov's and Pinsker's  \cite{gorbunov-pinsker1973}  nonanticipatory epsilon entropy of a process $X^n$ with distribution ${\bf P}_{X^n}$,   is defined by 
\begin{align}
&R_{X^n}(\epsilon)   \tri   \inf_{{\bf P}_{X^n, Y^n} }    \int\log\Big(\frac{ {\bf P}_{Y^n, X^n}}{{\bf P}_{Y^n} \times  { \bf P}_{Y^n}  }\Big)  {\bf P}_{X^n,Y^n}   \label{gp_1}  \\
&\mbox{subject to the average disrortion (\ref{dis_1}) and  }  \\
& \mbox{causality,  ${\bf P}_{Y^t|X^n} = {\bf P}_{Y^t|X^t},  t=1, \ldots n$}\label{gp_2}
\end{align}
where the infimum is taken over all joint distributions ${\bf P}_{X^n, Y^n}$ such that the ${\mathbb X}$ marginal distribution is the fixed distribution ${\bf P}_{X^n}$, and the fidelity and causality are satisfied. 
Shannon's RDF corresponds to $R_{X^n}(\epsilon)$ without the causality restriction  (\ref{gp_2}).   ${\bf P}_{Y^n|X^n}$ and ${\bf P}_{X^n|Y^n}$  are  known as,  the forward test channel  and  the backward test channel, respectively,    of reconstructing $X^n$ by $Y^n$ subject to fidelity (\ref{dis_1}). \\
Over the years,  $R_{X^n}(\epsilon)$ is applied in  the following areas.

\begin{list4}
\item[1)] Quantification of the rate loss of the OPTA by causal codes \cite{neuhoff:1982} and zero-delay codes \cite{gaarder-slepian1982,linder-lagosi2001}, with respect to noncausal codes,  for Gaussian Markov processes $X^n$  with square-error fidelity  \cite{derpich:2012}.  The construction of  causal and zero-delay codes \cite{stavrou-ostergaard-charalambos:2018}, based on subtractive dither with uniform scalar quantization (SDUSQ) \cite{zamir-feder1996}. 

\item[2)] Necessary and sufficient conditions to stabilize unstable linear Gaussian control systems over finite rate, noiseless or noisy,  communications channels,  and to design controllers, encoders and decoders subject to finite rate constraints \cite{tatikonda:2004,nair-evans2004}.

\item[3)] Synthesize  recursive, causal filters of Gaussian Markov processes subject to square-error fidelity \cite{charalambous-stavrou-ahmed2014ieeetac,stavrou-charalambous-charalambous-loyka2018siam}. 
\end{list4}
However,  the complete characterization  of the multivariate Gaussian Markov process $X^n$ with square-error fidelity, i.e., the specification of the realization of $Y^n$ and its structural properties,  which induces the optimal test channel, was only recently completed in \cite{charalambous2020structural}, although the problem was posed and solved for the scalar Gaussian Markov source by Gorbunov and  Pinsker in \cite{gorbunov-pinsker1974}.

In this paper we  formulate and analyze  the  nonanticipatory epsilon entropy,  designated  henceforth  by the name joint nonanticipative rate distortion function (NRDF),    of a tuple of processes $(X_1^n, X_2^n)$,   when each process is assigned an individual fidelity criterion, as shown in Fig.~\ref{fig:jointRDF}.  

Our interest in this problem is motivated by the   classical joint compression problem of a tuple of jointly independent and identically distributed processes $(X_1^n, X_2^n)$ with individual fidelity criteria,  introduced in \cite{xiao}.  As pointed out in \cite{xiao}, contrary to the classical joint RDF of a tuple   process,   viewed as a single process,  $X^n=(X_1^n, X_2^n)$, with a  single fidelity criterion assigned to $X^n$, the classical  joint RDF with individual fidelity criteria    aims  at  the design of encoders with cooperation, and hence fundamentally different from the former.  Indeed,  inherent in the optimal test channel distribution is the encoder cooperation,  which  is absent in the classical joint RDF with a single fidelity criterion. The additional level of  complexity of the encoder cooperation  is   demonstrated  in   \cite{xiao},  through the calculation of the classical joint RDF of a tuple of scalar,  jointly independent and identically distributed Gaussian RVs with individual square-error distortion criteria. The  recent treatment  in \cite{stylianou2021joint},  of the classical joint RDF for a tuple of  multivariate jointly independent and identically distributed Gaussian processes with individual square-error distortion criteria,   demonstrated additional  challenges, which are due to  the consideration of the  multivariate analog of   \cite{xiao}, and which are absent in the classical water-filling solution of the analogous RDF with a single fidelity criterion. Another application of the  classical joint RDF of a tuple of processes with individual fidelity criteria is the  Gray and Wyner \cite{gray1974source} source coding for a simple network, where this joint RDF is needed to characterize the rate region \cite{charalambous2020characterization}.

\begin{figure}[t]
  \centering
  \includegraphics[width=0.99\columnwidth]{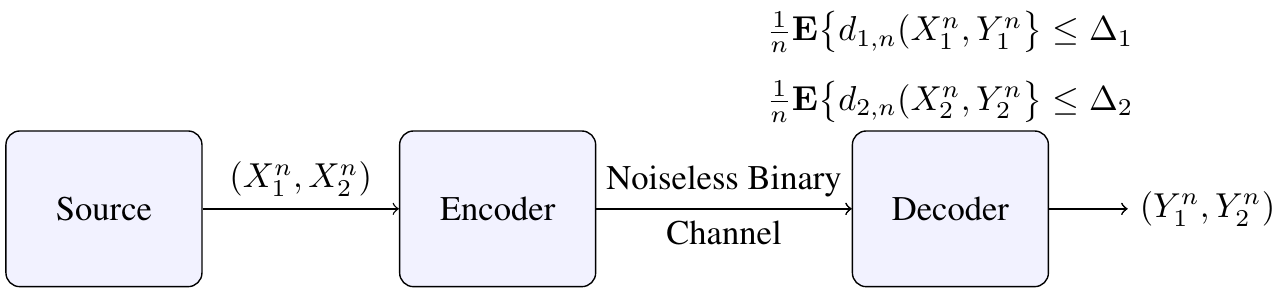}
\caption{Lossy Compression of correlated sources with individual distortion criteria.}
  \label{fig:jointRDF}
\end{figure}

\subsection{Problem Statement and Main Results}

We  consider a tuple of random processes   $X_i^n \tri \{X_{i,1}, \ldots, X_{i,n}\}, i=1,2$,
\begin{align}
X_{i,t} : \Omega \rar {\mathbb X}_i, \hso  t=1,\ldots, n, \hso  i=1,2
\end{align}
where ${\mathbb X}_i, i=1,2$ are metric spaces, 
 with corresponding tuple of reproduction processes  $Y_i^n \tri \{Y_{i,1}, \ldots, Y_{i,n}\}, i=1,2$, 
\begin{align} 
 Y_{i,t} : \Omega \rar {\mathbb Y}_i, \hso  t=1,\ldots, n, \hso i=1,2
 \end{align} 
where ${\mathbb Y}_i \subseteq {\mathbb X}_i, i=1,2$. The reproduction tuple satisfies  two individual fidelity criteria of reconstructing $X_i^n$ by $Y_i^n$, $i=1,2$, defined by  the measurable functions $d_{i,n}: {\mathbb X}_i^n \times {\mathbb Y}_i^n \rar [0,\infty), i=1,2$,
\begin{align}
&\frac{1}{n} {\bf E} \Big\{d_{i,n}(X_i^n, Y_i^n)  \Big\}\leq \Delta_i,  \hst i=1,2,  \label{AR_M_4}\\
 & d_{i,n}(x_i^n, x_i^n)=  \sum_{t=1}^n \rho_t(x_t, y_t), \hso    i=1,2.
\end{align}
\noi Given the fixed joint distribution ${\bf P}_{X_1^n, X_2^n}$ of the tuple $(X_1^n, X_2^n)$,  we define the joint distribution of $(X_1^n, X_2^n, Y_1^n, Y_2^n)$,  using the forward  ${\bf P}_{Y_1^n,Y_2^n|X_1^n,X_2^n}$ and  backward  ${\bf P}_{X_1^n,X_2^n|Y_1^n,Y_2^n}$  test channel distributions     by 
\begin{align}
{\bf P}_{X_1^n,X_2^n,Y_1^n,Y_2^n}=&{\bf P}_{Y_1^n,Y_2^n|X_1^n,X_2^n}\otimes{\bf P}_{X_1^n, X_2^n} \label{jdf}\\
=&{\bf P}_{X_1^n,X_2^n|Y_1^n,Y_2^n}\otimes{\bf P}_{Y_1^n, Y_2^n}
\end{align}
where  $\otimes$ denotes  the compound probability operator. \\
Let  $I(X_1^nX_2^n; Y_1^n,Y_2^n)$ denote  the mutual information between the tuple $(X_1^n, X_2^n)$ and its reproduction tuple  $(Y_1^n, Y_2^n)$,   defined by \cite{berger:1971}
\begin{align}
I(X_1^n, X_2^n;&Y_1^n,Y_2^n) \tri  \int\log\Big(\frac{ {\bf P}_{X_1^n,X_2^n, Y_1^n,Y_2^n}}{{\bf P}_{X_1^n,X_2^n}\times{\bf P}_{Y_1^n,Y_2^n}}\Big) \nonumber \\
& {\bf P}_{X_1^n,X_2^n,Y_1^n,Y_2^n} .    \label{MI_1_a} 
\end{align}
Let  ${\cal Q}_{X_1^n, X_2^n}^S(\Delta_1, \Delta_2)$ denote the fidelity constraint set of the two individual distortions,   defined by 
\begin{align}
&{\cal Q}_{X_1^n, X_2^n}^S(\Delta_1, \Delta_2)\tri   \Big\{{\bf P}_{Y_1^n,Y_2^n|X_1^n,X_2^n} \Big| \mbox{\: the ${\mathbb X}_1^n \times {\mathbb X}_2^n-$  }    \nonumber \\
&\mbox{marginal of the joint dist. (\ref{jdf}) is  ${\bf P}_{X_1^n, X_2^n}$, and  } \nonumber \\
&\hst  \frac{1}{n}{\bf E}\Big\{d_{i,n}(X_i^n,Y_i^n)\Big\}\leq{\Delta_i}, \: i=1,2 \Big\} .
\end{align} 

\noi {\it The  joint NRDF  for the tuple  $(X_1^n, X_2^n)$ with individual fidelity criteria} is defined by 
\begin{align}R_{X_1^n, X_2^n}(\Delta_1, \Delta_2)\tri& \inf_{{\cal Q}_{X_1^n, X_2^n}^S(\Delta_1, \Delta_2): \; \mbox{(C) holds} } \Big\{ \nonumber \\
& I(X_1^n, X_2^n;Y_1^n,Y_2^n)\Big\}  \label{eq.1_pw}
\end{align}
where the infimum is taken over all joint distributions 
${\bf P}_{X_1^n,X_2^n,Y_1^n,Y_2^n}\in{\cal Q}_{X_1^n, X_2^n}^S(\Delta_1, \Delta_2) $ such that  following condition holds:\\


\noi {\bf (C)} for each $t\in \{1,\ldots, n\}$, the process $(Y_1^t, Y_2^t)$ is conditionally independent of $(X_{1,t+1}^n, X_{2,t+1}^2)$ conditioned on $(X_1^t,X_2^t)$, that is,  
\begin{align}
{\bf P}_{Y_1^t,Y_2^t|X_1^n,X_2^n}={\bf P}_{Y_1^t, Y_2^t|X_1^t,X_2^t}, \hso t=1, \ldots, n-1 \label{CI_1}
\end{align}
 equivalently expressed as a Markov chain ($\leftrightarrow$) 
\begin{align}
(X_{1,t+1}^n, X_{2,t+1}^n) \leftrightarrow {(}X_1^t,X_2^t) \leftrightarrow (Y_1^t,Y_2^t), \; t=1, \ldots, n-1. \nonumber 
\end{align}
\noi {\it Conditional independence  (\ref{CI_1})} is a   ``causality condition'' of the reproduction distribution.\\
{The main result of  the first part of the paper, is  }

(R1)  the structural properties of optimal test channel distributions, and realizations of the reproduction processes. 
 
 In the second part of the paper, we analyze   the joint NRDF ${R}_{X_1^n,X_2^n}(\Delta_1,\Delta_2)$, for a tuple of jointly  multivariate Gaussian Markov processes  $(X_1^n, X_2^n)$,  and  two  square-error distortion  functions,  defined by 
\begin{align}
&X_{i,t}: \Omega  \rar {\mathbb X}_i \tri    {\mathbb R}^{p_i},  \hso t=1, \ldots, n,\label{prob_1}
 \\
 & {\bf P}_{X_{1,t},X_{2,t}|X_1^{t-1}, X_2^{t-1}}={\bf P}_{X_{1,t},X_{2,t}|X_{1,t-1}, X_{2, t-1}}, \\
& (X_{1,t}, X_{2,t}) \in G(0, Q_{(X_{1,t}, X_{2,t}})), \label{prob_2}\\
&Q_{(X_{1,t}, X_{2,t})} \tri   {\mathbf E} \bigg\{ \begin{pmatrix} X_{1,t} \\ X_{2,t}  \end{pmatrix}  \begin{pmatrix}X_{1,t} \\ X_{2,t}  \end{pmatrix}\T \bigg\}\\
& Y_{i,t}: \Omega  \rar  {\mathbb Y}_i \tri    {\mathbb R}^{p_i}, \hso i=1,2,\label{prob_8} \\
& d_{i,n} (x_i^n, y_i^n)= \frac{1}{n} \sum_{t=1}^n ||x_{i,t}-y_{i,t}||_{{\mathbb R}^{p_{i}}}^2,\;\;\;i=1,2.  \label{prob_9}
\end{align}
Here  
$X \in G(0,Q_X)$ means $X$ is a Gaussian RV,  with zero mean and 
 covariance matrix $Q_X\succeq 0$.\\
Our main contributions include,  

(R2) realizations of optimal reproduction process $(Y_2^n, Y_2^n)$, and its structural properties, and 

(R3) characterization of  joint NRDF ${R}_{X_1^n,X_2^n}(\Delta_1,\Delta_2)$.

\noi

%
%
%
%

%
%
%
%
\section{Joint Nonanticipative RDF with Individual Fidelity Criteria} \label{sec:ISofNRDF}

\subsection{Notation}\label{sec:preliminaries}
$\mathbb{R} \triangleq (-\infty,\infty)$, $\mathbb{Z} \tri \{\ldots,-1,0,1,\ldots\}$, {$\mathbb{Z}_0 \triangleq \{0,1,2, \ldots\}$,} $\mathbb{N}\triangleq\{1,2, \ldots\}$, $\mathbb{N}^n\triangleq\{1,\ldots,n\}$, $n \in {\mathbb N}$. For any  matrix $A\in \mathbb{R}^{p\times m}, (p,m)\in {\mathbb N}\times {\mathbb N}$, we denote its transpose by $A\T$,   its pseudoinverse by  $A^\dagger \in \mathbb{R}^{m\times p}$,  and for $m=p$,  we denote its trace  by  $\trace(A)$. The $n$ by $n$ identity (resp. zero) matrix is represented by $I_n$ (resp. $0_n$). 
 ${\cal S}_+^{p\times p}$ denotes the set of symmetric positive semidefinite  matrices  $A\in \mathbb{R}^{p\times p}$, and ${\cal S}_{++}^{p\times p}$ its subset of positive definite matrices. The statement $A\succeq B$ (resp. $A\succ B$) means that $A-B$ is symmetric positive semidefinite (resp. positive definite).   \\
Denote an arbitrary set or space by ${\mathbb U}$ and the product space formed by $n \in {\mathbb N}$ copies of it by  ${\mathbb U}^n \tri \times_{t=1}^n {\mathbb U}$. $u^n \in {\mathbb U}^n$ denotes the set of $n-$tuples $u^n \tri (u_1,u_2, \ldots, u_n)$, where $u_k \in {\mathbb  U}, k=1, \ldots, n$ are its coordinates.\\
Denote a probability space by $(\Omega, {\cal F}, {\mathbb P})$. For a sub-sigma-field  ${\cal  G} \subseteq {\cal  F}$,  $A \in {\cal  F}$, we denote by ${\mathbb P}(A|{\cal  G})={\mathbb P}(A|{\cal  G})(\omega), \omega \in \Omega$ the conditional probability of $A$ given ${\cal  G}$.
 For a tuple of real-valued RVs (RV) $X: \Omega \rightarrow {\mathbb X}, Y: \Omega \rightarrow {\mathbb Y}$, where $({\mathbb X}, {\cal  B}({\mathbb  X})),    ({\mathbb Y}, {\cal B}({\cal Y}))$ are  measurable spaces,  we denote the measure (resp.  joint distribution, if ${\mathbb X}, {\mathbb Y}$ are Euclidean spaces) induced by RVs $(X, Y)$ on ${\mathbb X} \times {\mathbb Y}$ by  ${\bf P}(dx,dy)$ (resp. ${\bf P}_{X,Y}$), and their marginals on ${\mathbb X}$ and ${\mathbb Y}$ by  ${\bf P}(dx)$ and ${\bf P}(dy)$  (resp.  ${\bf P}_{X}$ and ${\bf P}_{Y}$), respectively. We denote the conditional distribution of RV $X$ conditioned on $Y$ by ${\bf P}_{X|Y}$ or ${\bf P}(dx|y)$, if $Y=y$ is fixed.  \\
For  a triple of real-valued RVs $X: \Omega \rightarrow {\mathbb X}, Y: \Omega \rightarrow {\mathbb Y}$, $Z: \Omega \rightarrow {\mathbb Z}$,  we say that RVs $(Y, Z)$ are conditional independent given RV $X$ if ${\bf P}_{Y, Z|X}={\bf P}_{Y|X} {\bf P}_{Z|X}-$a.s (almost surely) or equivalently ${\bf P}_{Z|X,Y}={\bf P}_{Z|X}-$a.s; the specification a.s is often omitted.  We often  denote the above conditional independence by the Markov chain (MC) $Y \leftrightarrow X \leftrightarrow Z$.\\
The conditional covariance of the two-component vector RV $X = (X_1\T, X_2\T)\T$, $X_i: \Omega \rar {\mathbb R}^{p_i}, i=1,2$ conditioned on the two-component  vector $Y = (Y_1\T, Y_2\T )\T$, $Y_i: \Omega \rar {\mathbb R}^{p_i}, i=1,2$ is denoted by $Q_{(X_1,X_2)|Y}\tri \Cov\Big(X,X\Big|Y\Big) \succeq 0$, where
\begin{align}
Q_{(X_1,X_2)|Y} =& \begin{pmatrix}
Q_{X_1|Y} & Q_{X_1,X_2|Y}  \\
Q_{X_1,X_2|Y}\T & Q_{X_2|Y} \\
\end{pmatrix} \in {\mathbb R}^{(p_1+p_2)\times (p_1+p_2)}, \nonumber\\
Q_{X_1,X_2|Y}  \tri &\Cov\Big(X_1,X_2\Big| Y\Big). \nonumber \\
 \sr{(1)}{=} &  \mathbf{E} \Big\{ \Big( X_1 - {\mathbf E}\Big\{ X_1\Big|Y\Big\} \Big) \Big(X_2 - {\mathbf E}\Big \{ X_2\Big|Y\Big\} \Big)\T \Big\}  \nonumber \\
= &  \mathbf{E} \Big\{ E_1  E_2\T \Big\}, \hso E_i \tri X_i - {\mathbf E}\Big\{ X_i\Big|Y\Big\}, \hso i=1,2 \nonumber\\
\equiv & \Sigma_{E_1, E_2}
\end{align}
and where (1) holds if $(X_1, X_2, Y_1, Y_2)$ is jointly Gaussian. Similarly for $Q_{X_i|Y}, i=1,2$. Consequently, for jointly Gaussian RVs $(X_1, X_2, Y_1, Y_2)$, and the  two-component vector  RV $E  \tri (E_1\T,E_2\T)\T$,  we have $Q_{(X_1,X_2)|Y}=\Sigma_{(E_1,E_2)}$.

\ \

\subsection{Equivalent Sequential Formula of Joint NRDF}
First, we give the sequential equivalent of  the joint NRDF $R_{X_1^n,X_2^n}(\Delta_1, \Delta_2)$. We make use of the following lemma.\\
\begin{lemma} \cite{charalambous-stavrou-ahmed2014ieeetac,stavrou-charalambous-charalambous-loyka2018siam} Conditional independence conditions\\
\label{equivalent_statements}
 The  following statements are equivalent  $\forall n \in {\mathbb N}$.
\begin{list5}
\item[{ MC1}:] ${\bf P}_{Y_1^n,Y_2^n|X_1^n, X_2^n}=\otimes_{t=1}^n{\bf P}_{Y_{1,t}, Y_{2,t}|Y_1^{t-1},Y_2^{t-1},X_1^t, X_2^t}$.

\item[{ MC2}:]  $(Y_{1,t}, Y_{2,t}) \leftrightarrow (X_1^t,X_2^t, Y_1^{t-1}, Y_2^{t-1}) \leftrightarrow (X_{1,t+1}^n, X_{2,t+1}^n)$ forms a MC,~for each $t=1,\ldots, n-1$.

\item[{ MC3}:] $(Y_1^t, Y_2^t) \leftrightarrow (X_1^t, X_2^t) \leftrightarrow (X_{1,t+1}, X_{2,t+1})$ forms a MC,~for each  $t=1,\ldots, n-1$.

\item[{ MC4}:] $(X_{1,t+1}^n, X_{2, t+1}^n)\leftrightarrow (X_1^t,X_2^t)\leftrightarrow (Y_1^t, Y_2^t)$ forms a MC,~for each $t=1,\ldots, n-1$. 
\end{list5}
\end{lemma}

\ \

By Lemma~\ref{equivalent_statements}, Condition (C)  is equivalent to MC1, and the joint distribution of $(X_1^n, X_2^n, Y_1^n, Y_2^n)$, is expressed as 
\begin{align}
&{\bf P}_{X_1^n,X_2^n,Y_1^n,Y_2^n}={\bf P}_{Y_{1,n}, Y_{2,n}|Y_1^{n-1}, Y_2^{n-1},X_1^n,X_2^n}\nonumber \\
&\hst  \otimes {\bf P}_{X_{1,n}, X_{2,n}|Y_1^{n-1}, Y_2^{n-1},X_1^{n-1},X_2^{n-1}} \ldots \nonumber \\
&\hst \otimes  {\bf P}_{Y_{1,2}, Y_{2,2}|Y_{1,1}, Y_{2,1},X_1^2,X_2^2} \otimes    {\bf P}_{X_{1,2}, X_{2,2}|X_{1,1}, X_{2,1}}\nonumber \\
&\hst \otimes  {\bf P}_{Y_{1,1}, Y_{2,1}|X_{1,1},X_{2,1}} \otimes    {\bf P}_{X_{1,1}, X_{2,1}}\ \label{sjdf}
\end{align}
The information measure $I(X_1^n, X_2^n;Y_1^n,Y_2^n)$ in (\ref{eq.1_pw})  is expressed sequentially as, 
\begin{align}
&I(X_1^n, X_2^n;Y_1^n,Y_2^n)\nonumber \\
&= {\bf E}\Big\{ \sum_{t=1}^n   \log \Big( \frac{{\bf P}_{Y_{1,t}, Y_{2,t}|Y_1^{t-1},Y_2^{t-1},X_1^t,X_2^t}}{{\bf P }_{Y_{1,t},Y_{2,t} |Y_1^{t-1}, Y_2^{t-1}} }\Big)  \Big\} \nonumber \\
&= \sum_{t=1}^n I(X_1^t, X_2^t; Y_{1,t}, Y_{2,t}|Y_1^{t-1}, Y_2^{t-1}).\nonumber 
\end{align}
 The  joint NRDF  
$R_{X_1^n,X_2^n}(\Delta_1, \Delta_2)$ of (\ref{eq.1_pw}) subject to condition {(C)} is expressed sequentially as follows. 
\begin{align}
&R_{X_1^n, X_2^n}(\Delta_1, \Delta_2) \label{NRDF_A}  \\
&= \inf_{{\cal Q}_{X_1^n,X_2^n}(\Delta_1, \Delta_2)}\Big\{\sum_{t=1}^n I(X_1^t, X_2^t; Y_{1,t}, Y_{2,t}|Y_1^{t-1}, Y_2^{t-1}) \Big\} \nonumber
\end{align}
where 
\begin{align}
&{\cal Q}_{X_1^n, X_2^n}(\Delta_1, \Delta_2)\triangleq  \Big\{{\bf P}_{ X_1^t,X_2^t, Y_1^t, Y_2^t}, t=1,\ldots, n \Big|\nonumber \\
&\hst  \mbox{ (\ref{sjdf}) holds,   the ${\mathbb X}_1^n \times {\mathbb X}_2^n-$ marginal  is  ${\bf P}_{X_1^n, X_2^n}$, }    \nonumber \\
&\hst \mbox{ }\: \frac{1}{n}{\bf E}\Big\{d_{i,n}(X_i^n,Y_i^n)\Big\}\leq{\Delta_i}, i=1,2\Big\}.
\end{align}
It can be shown that 
\begin{align}
&{\cal Q}_{X_1^n, X_2^n}(\Delta_1, \Delta_2)= \Big\{{\bf P}_{Y_{1,t} Y_{2,t}|Y_1^{t-1},Y_2^{t-1}, X_1^t,X_2^t}, t=1,\ldots, n \Big|\nonumber \\
&\hst  \mbox{ the ${\mathbb X}_1^n \times {\mathbb X}_2^n-$ marginal of   (\ref{sjdf}) is  ${\bf P}_{X_1^n, X_2^n}$, }    \nonumber \\
&\hst \frac{1}{n}{\bf E}\Big\{d_{i,n}(X_i^n,Y_i^n)\Big\}\leq{\Delta_i}, i=1,2\Big\}.  
  \label{NRDF_D_A}
\end{align}

\subsection{Information Structures of  Sequential  Joint RDF for a Tuple of Markov Processes}
The main result of this section is Theorem~\ref{thm:jointlower}, which identifies structural properties of  the realizations $(Y_1^n, Y_2^n)$, of the test channels that minimize $\sum_{t=1}^nI(X_1^t,X_2^t;Y_{1,t},Y_{2,t}|Y_1^{t-1}, Y_2^{t-1})$, when the joint process $(X_1^n, X_2^n)$ is Markov and the fidelity is defined with respect to the square-error.

First,  we recall a preliminary result, of  a structural property of    test channel distributions from the set ${\cal Q}_{X_1^n, X_2^n}(\Delta_1, \Delta_2)$. \\

\begin{theorem} \cite{charalambous-stavrou2014ecc,stavrou-charalambous-charalambous-loyka2018siam}\\
\label{IS_general}
Consider the joint NRDF of (\ref{NRDF_A}), and assume 
the joint process $(X_1^n, X_2^n)$ is  Markov,  that is, ${\forall t \in {\mathbb N}^n}$
\begin{align}
{\bf P}_{X_{1,t}, X_{2,t}|X_1^{t-1}, X_2^{t-1}}= {\bf P}_{X_{1,t}, X_{2,t}|X_{1,t-1}, X_{2,t-1}}. \label{ms_1}
\end{align}
Then the joint NRDF is given by 
\begin{align}
&R_{X_1^n,X_2^n}(\Delta_1, \Delta_2)=\inf_{ {\cal M}_{X_1^n,X_2^n}(\Delta_1,\Delta_2)} \Big\{   \label{NRDF_Markov_1_a} \\
& {\bf E}\Big\{ \sum_{t=1}^n   \log \Big( \frac{{\bf P}_{Y_{1,t}, Y_{2,t}|Y_1^{t-1},Y_2^{t-1},X_{1,t},X_{2,t}}}{{\bf P }_{Y_{1,t},Y_{2,t} |Y_1^{t-1}, Y_2^{t-1}} }\Big)  \Big\} \Big\} \label{NRDF_Markov_1_b}\\
=&\inf_{{\cal M}_{X_1^n,X_2^n}(\Delta_1, \Delta_2)}\Big\{ \sum_{t=1}^n I(X_{1,t}, X_{2,t}; Y_{1,t}, Y_{2,t}|Y_1^{t-1}, Y_2^{t-1})\Big\}  \nonumber
\end{align}
where 
\begin{align}
&{\cal M}_{X_1^n,X_2^n}(\Delta_1, \Delta_2)\triangleq  \Big\{{\bf P}_{Y_{1,t} Y_{2,t}|Y_1^{t-1},Y_2^{t-1}, X_{1,t},X_{2,t}}, \nonumber \\
&\hso t=1,\ldots,n  \Big| \mbox{ the ${\mathbb X}_1^n \times {\mathbb X}_2^n-$ marginal corresp. to  (\ref{ms_1})}    \nonumber \\
&\hso  \frac{1}{n}{\bf E}\Big\{d_{i,n}(X_i^n,Y_i^n)\Big\}\leq{\Delta_i}, i=1,2\Big\},  \label{con}\\
& {\bf P}_{X_1^t,X_2^t, Y_1^t, Y_2^t}={\bf P}_{X_{1,1},X_{2,1}}  \otimes {\bf P}_{Y_{1,1}, Y_{2,1}|X_{1,1}, X_{2,1}}\nonumber  \\ 
& \otimes_{i=1}^t \Big( {\bf P}_{X_{1,t}, X_{2,t}|X_1^{i-1}, X_2^{i-1}}   \otimes   {\bf P}_{Y_{1,i}, Y_{2,i}|Y_1^{i-1},Y_2^{i-1}, X_1^t, X_2^t}  \Big), \nonumber \\
&{\bf P}_{Y_{1,t}, Y_{2,t}|Y_1^{t-1}, Y_2^{t-1}} = \int_{ {\mathbb X}_1 \times {\mathbb X}_2}  {\bf P}_{Y_{1,t}, Y_{2,t}|Y_1^{t-1},Y_2^{t-1}, X_{1,t} X_{2,t}} \nonumber \\
& \hspace*{4.0cm} \otimes {\bf P}_{X_{1,t}, X_{2,t}|Y_1^{t-1}, Y_2^{t-1}} 
\end{align}
\end{theorem}

\ \

Next,  we identify an important structural property of the optimal reproduction process $(Y_1^n, Y_2^n)$. \\

\begin{theorem} Structural property of  reproduction process\\
 \label{thm:jointlower}
Consider the statement of Theorem~\ref{IS_general} and the joint NRDF $R_{X_1^n,X_2^n}(\Delta_1, \Delta_2)$ of  (\ref{NRDF_Markov_1_a}) for  the Markov $(X_1^n, X_2^n)$. 

\noi (a) Define 
\begin{align}
&\widehat{X}_{i,t}= g_{i,t}(Y_1^t, Y_2^t), \; \forall t \in {\mathbb N}^n, \:   i=1,2, \\
&g_{i,t}: {\mathbb Y}_1^t\times{\mathbb Y}_2^t \rightarrow {\mathbb Y}_i,\hso \mbox{$g_{i,t}(\cdot)$ are meas. functions,} \; i=1,2 \nonumber.
\end{align}
Then, the following inequality holds for $t=1, \ldots, n$: 
\begin{align}
&I(X_{1,t},X_{2,t};Y_{1,t},Y_{2,t}|Y_1^{t-1}, Y_2^{t-1}) \nonumber \\
 &\quad \geq I(X_{1,t},X_{2,t}; {\widehat{X}_{1,t}},{\widehat{X}_{2,t}}|Y_1^{t-1}, Y_2^{t-1}).\label{eq:MutualInfoIneq}
\end{align}
Moreover,  if there exist  $(\widehat{X}_{1,t}, \widehat{X}_{2,t})$ such that  the functions $g_{i,t}(\cdot,\cdot)$ satisfy for $ i=1,2  $
\begin{align}
g_{i,t} (Y_1^t,Y_2^t)  =  \mathbf{E} \Big\{ X_{i,t}\Big|Y_1^t, Y_2^t \Big\}=Y_{i,t}, \; \forall t \in {\mathbb N}^n,  \label{eq_cm}
\end{align}
then the inequality in \eqref{eq:MutualInfoIneq} holds with equality.\\
(b) Let  ${\mathbb X}_1 \times {\mathbb X}_2\times {{\mathbb Y}}_1\times {{\mathbb Y}}_2={\mathbb R}^{p_1} \times {\mathbb R}^{p_2}\times {\mathbb R}^{p_1}\times {\mathbb R}^{p_2}$, $(p_1,p_2) \in {\mathbb Z}_+$. 
 For all measurable functions $h_{i,t}(Y_1^t,Y_2^t)$, $i=1,2$ then 
\begin{align}
&{\bf E}\Big\{\big|\big|X_{i,t}-h_{i,t}(Y_1^t,Y_2^t)    \big|\big|_{{\mathbb R}^{p_i}}^2\Big\}  \nonumber \\
&   \geq {\bf E}\Big\{\big|\big|X_{i,t}-\mathbf{E} \Big\{ X_{i,t}\Big| Y_1^t,Y_2^t   \Big\}\big|\big|_{{\mathbb R}^{p_i}}^2\Big\},\hso \forall t \in {\mathbb N}^n, \: i=1,2.\nonumber
\end{align}
(c) Suppose ${\mathbb X}_1 \times {\mathbb X}_2\times {{\mathbb Y}}_1\times {{\mathbb Y}}_2={\mathbb R}^{p_1} \times {\mathbb R}^{p_2}\times {\mathbb R}^{p_1}\times {\mathbb R}^{p_2}$, $(p_1,p_2) \in {\mathbb Z}_+$, and (\ref{eq_cm})    holds. Then the joint RDF  given by  (\ref{NRDF_Markov_1_b}) is characterized by 
\begin{align}
R_{X_1^n,X_2^n}&(\Delta_1, \Delta_2) =\inf_{{\cal M}_{X_1^n, X_2^n}^{cm}(\Delta_1, \Delta_2)}\Big\{ \nonumber \\
 &\sum_{t=1}^n I(X_{1,t}, X_{2,t}; Y_{1,t}, Y_{2,t}|Y_1^{t-1}, Y_2^{t-1})\Big\} \label{jRDF_g_cm}
\end{align}
where  $ {\cal M}_{X_1^n, X_2^n}^{cm}(\Delta_1, \Delta_2) \subseteq {\cal M}_{X_1^n, X_2^n}(\Delta_1, \Delta_2)$, with the additional restriction   ${g}_{i,t} (Y_1^t,Y_2^t)=\mathbf{E} \Big\{ X_{i,t}\Big|Y_1^t, Y_2^t\Big\}=Y_{i,t},-a.s. , \forall t\in {\mathbb N}^n$ for $i=1,2$.
\end{theorem}
\begin{proof} (a) By properties of mutual information, we have 
\begin{align}
 & I(X_{1,t},X_{2,t};Y_{1,t},Y_{2,t}|Y_1^{t-1}, Y_2^{t-1})\\
&\sr{(1)}{=}  I(X_{1,t},X_{2,t};Y_{1,t},Y_{2,t},\widehat{X}_{1,t},\widehat{X}_{2,t} |Y_1^{t-1}, Y_2^{t-1})\\ 
&\sr{(2)}{=} I(X_{1,t},X_{2,t}; Y_{1,t},Y_{2,t}| \widehat{X}_{1,t},\widehat{X}_{2,t}, Y_1^{t-1}, Y_2^{t-1}) \nonumber \\&\quad + I(X_{1,t},X_{2,t}; \widehat{X}_{1,t},\widehat{X}_{2,t}| Y_1^{t-1}, Y_2^{t-1}  )\nonumber \\
& \sr{(3)}{\geq}I(X_{1,t},X_{2,t}; \widehat{X}_{1,t},\widehat{X}_{2,t}| Y_1^{t-1}, Y_2^{t-1}  ), \label{eq:LB}
\end{align} 
where \((1)\) is due to $\widehat{X}_i,i=1,2$,  are  functions of $(Y_1^t,Y_2^t)$,  \((2)\) is due to the chain rule of mutual information, and \((3)\)   is due to $ I(X_{1,t},X_{2,t}; Y_{1,t},Y_{2,t}| \widehat{X}_{1,t},\widehat{X}_{2,t}, Y_1^{t-1}, Y_2^{t-1})\geq 0$.  Thus,  (\ref{eq:MutualInfoIneq}) is obtained. Furthermore, if  $\widehat{X}_{i,t}=g_{i,t}(Y_1^t,Y_2^t)  = Y_{i,t}-a.s, i=1,2$  hold,   then $I(X_{1,t},X_{2,t}; Y_{1,t},Y_{2,t}| \widehat{X}_{1,t},\widehat{X}_{2,t}, Y_1^{t-1}, Y_2^{t-1})=0$, and hence the inequality (\ref{eq:LB}) become equality.  
(\textit{b}) The inequality is well-known, due to the orthogonal projection theorem. 
(c) This is due to (a), (b), and the fact that the fidelity  constraints hold with equality.
\end{proof}

\ \

\begin{remark} For a tuple of  Gaussian Markov processes, $(X_1^n, X_2^n)$,  Theorem~\ref{IS_general} and      Theorem~\ref{thm:jointlower}, are used in the remaining paper to characterize  joint NRDF.
\end{remark}

\section{Joint NRDF of Multivariate Gaussian Markov Processes with  Individual MSE Distortion Criteria} 
\label{ar_1}
For the rest of the paper we consider the tuple of multivariate Gaussian Markov process of Definition~\ref{def:gmp}.\\

\begin{definition}
\label{def:gmp}
A  tuple of  multivariate  Gaussian Markov process, $X_t = (X_{1,t}\T, X_{2,t}\T)\T$,  $X_{i,t}: \Omega \rar {\mathbb R}^{p_i}, t=0, \ldots, n$, $i=1,2$,  is defined for ${t=1, \ldots, n-1}$  by the recursion
\begin{align}
\begin{pmatrix}
X_{1,t+1} \\X_{2,t+1}
\end{pmatrix}=A_{t} \begin{pmatrix}
X_{1,t} \\X_{2,t}
\end{pmatrix}+B_{t} \begin{pmatrix}
W_{1,t+1} \\W_{2,t+1}
\end{pmatrix},  \:  \label{SSM_1}
\end{align}
where \emph{(i)}  $A_{t}\in\mathbb{R}^{(p_1+p_2)\times (p_1+p_2)}, B_{t}\in\mathbb{R}^{(p_1+p_2)\times{(q_1 +q_2)}}$ are non-random matrices; \emph{(ii)} $\{W_{i,t}: t=2, \ldots, n-1\}$ is an $\mathbb{R}^{{q_i}}$-valued independent  Gaussian process,  for $i=1,2$, $W_t = (W_{1,t}\T, W_{2,t}\T)\T \in G(0, Q_{(W_{1,t}, W_{2,t})})$, $Q_{(W_{1,t}, W_{2,t})}\succeq 0$, independent of $X_1$; \emph{(iii)} $X_1\in\mathbb{R}^{p_1+p_2}$ is Gaussian $X_1 \in G(0, Q_{(X_{1,1}, X_{2,1})}), Q_{(X_{1,1}, X_{2,1})}\succeq 0$.
\label{def_GAR} 
\end{definition}

\ \

\begin{definition}
Define
\begin{align*}
&X_t \tri \begin{pmatrix}
X_{1,t} \\X_{2,t}
\end{pmatrix},\; Y_t \tri  \begin{pmatrix}
Y_{1,t} \\ Y_{2,t}
\end{pmatrix}, \:  \hso \forall t \in {\mathbb N}^n, \\   
&E_t \tri  \begin{pmatrix}
E_{1,t}\\ E_{2,t}
\end{pmatrix}, \: E_t^- \tri  \begin{pmatrix}
E_{1,t}^-\\ E_{2,t}^-
\end{pmatrix}, \nonumber \\ 
& E_{i,t}\tri X_{i,t} - \widehat{X}_{i,t|t}, \; E_{i,t}^-\tri X_{i,t} - \widehat{X}_{i,t|t-1}, i=1,2, \\
&\widehat{X}_{t|s} \tri   \mathbf{E}\Big\{ \begin{pmatrix}
X_{1,t} \\ X_{2,t} 
\end{pmatrix} \Big| Y^s  \Big\}= \begin{pmatrix} \widehat{X}_{1,t|s} \\ \widehat{X}_{2,t|s}\end{pmatrix}, \; \forall (t,s)\in {\mathbb N}^n \times {\mathbb N}^n \nonumber
\end{align*}
and  the mean-square errors  
\begin{align}
&\Sigma_{(E_{1,t}, E_{2,t})} \tri  {\bf E}\Big\{E_t E_t\T\Big\}, \ \ \forall{t} \in {\mathbb N}^n, \\
&\Sigma_{(E_{1,t}^-, E_{2,t}^-)} \tri  {\bf E}\left\{E_t^- \big(E_t^-\big)\T\right\}\ \ \forall {t} \in {\mathbb N}^n
\end{align}
where for $t=1$,   $\Sigma_{(E_{1,1}^-, E_{2,1}^-)}  \tri 
Q_{(X_{1,1}, X_{2,1})}$.
\end{definition}

Next, we present another structural property.
the tuple of Gaussian Markov process subject to two square-error distortion criteria. \\

\begin{theorem} \label{thm:suffjoint} 
Consider the joint NRDF $R_{X_1^n,X_2^n}(\Delta_1, \Delta_2)$ of (\ref{NRDF_A})   for the  tuple of  multivariate  Gaussian Markov process
of Definition~\ref{def:gmp}, with individual distortion criteria,  $d_{i,n}(x_1^n,y_1^n) \tri \frac{1}{n} \sum_{t=1}^n ||x_{1,t}-y_{{1},t}||_{{\mathbb R}^{p_i}}^2, i=1,2$. \\
The following hold.\\
(a) The  minimizing element of the  set ${\cal Q}_{X_1^n,X_2^n}(\Delta_1, \Delta_2)$ is  jointly Gaussian   ${\bf P}_{X_1^t,X_2^t, Y_1^t, Y_2^t}={\bf P}_{X_1^t,X_2^t, Y_1^t, Y_2^t}^G,  t=1, \ldots, n$,  and it is induced by the parametric realization 
 \begin{align}
\begin{pmatrix}
Y_{1,t} \\Y_{2,t}
\end{pmatrix}=&H_t \begin{pmatrix}
X_{1,t} \\X_{2,t}
\end{pmatrix}+\begin{pmatrix}  g_{1,t}( Y_1^{t-1}, Y_2^{t-1})  \\ g_{2,t}( Y_1^{t-1}, Y_2^{t-1}) \end{pmatrix} +\begin{pmatrix}
V_{1,t} \\V_{2,t}
\end{pmatrix}, \label{real_1_an} \\
 =&  H_tX_t+ \Big(I_{p_1+p_2}-H_t\Big)\widehat{X}_{t|t-1} + V_t
\end{align}
where 
\begin{align}
 &\text{$H_t =\begin{pmatrix}
H_{11,t} & H_{12,t}  \\H_{21,t} & H_{22,t}
\end{pmatrix}  \in {\mathbb R}^{(p_1+p_2)\times (p_1+p_2)}$ are nonrandom,} \label{real_1_b}  \\
 &g_{i,t}(Y_1^{t-1}, Y_2^{t-1})= \widehat{X}_{i,t|t-1} -  \begin{pmatrix}
H_{i1,t} & H_{i2,t} 
\end{pmatrix} \widehat{X}_{t|t-1}, \label{real_1_c} \\
&\widehat{X}_{t|t-1}= {\bf E}\Big\{X_t \Big|Y_1^{t-1}, Y_2^{t-1}\Big\} =A_{t-1} \begin{pmatrix} \widehat{X}_{1,t-1|t-1}  \\\widehat{X}_{2,t-1|t-1} \end{pmatrix}   \nonumber\\
 &\text{$V_t =(V_{1,t}\T, V_{2,t}\T)\T \in G(0, Q_{(V_{1,t}, V_{2,t})}), Q_{(V_{1,t}, V_{2,t})} \succeq 0$,} \nonumber\\
 &\text{ $V_t$ is indep.  of $X_1$ and $W_s=(W_{1,s}\T, W_{2,s}\T)\T$,} \: s = 1,  \ldots, t . \label{real_1_ane}
\end{align} 
Moreover,  \\
(i) the test channel, denoted by  ${\bf P}_{Y_{1,t}, Y_{2,t} |Y_1^{t-1}, Y_2^{t-1},X_{1,t}, X_{2,t}}^G$,   is  parametrized by $\big(H_t, Q_{(V_{1,t},V_{2,t})} \big), t=1,\ldots, n$, and satisfies 
\begin{align} 
&{\bf P}_{Y_{1,t}, Y_{2,t} |Y_1^{t-1}, Y_2^{t-1},X_{1,t}, X_{2,t}}^G \label{rep_11} \\
&=Q_t(dy_{1,t}, dy_{2,t}|\widehat{x}_{1,t-1|t-1},\widehat{x}_{2,t-1|t-1},  x_{1,t}, x_{2,t})\nonumber 
\end{align}
(ii)  for each $t=1, \ldots, n$ the pay-off satisfies
\begin{align}
&I(X_1^t, X_2^t; Y_{1,t}, Y_{2,t}|Y_1^{t-1}, Y_2^{t-1}) \label{NRDF_M_nn_N_aa}\\
&=I(X_{1,t}, X_{2,t}; Y_{1,t}, Y_{2,t}|Y_1^{t-1}, Y_2^{t-1}) \label{NRDF_M_nn_N_ab} \\
&= I(X_{1,t}, X_{2,t}; Y_{1,t}, Y_{2,t}|Y_1^{t-1}, Y_2^{t-1},\nonumber \\
& \hst \widehat{X}_{1,t-1|t-1}, \widehat{X}_{2,t-1|t-1}  ), \label{NRDF_M_nn_N_ac}
\\
&= I(X_{1,t}, X_{2,t}; Y_{1,t}, Y_{2,t}, \widehat{X}_{1,t|t}, \widehat{X}_{2,t|t} |Y_1^{t-1}, Y_2^{t-1},\nonumber \\
 &\hst \widehat{X}_{1,t|t-1}, \widehat{X}_{2,t|t-1}  ),\label{NRDF_M_nn_N_add}\\
 &= I(X_{1,t}, X_{2,t}; Y_{1,t}, Y_{2,t} |Y_1^{t-1}, Y_2^{t-1},\nonumber \\
 & \hst \widehat{X}_{1,t-1|t-1}, \widehat{X}_{2,t-1|t-1}, \widehat{X}_{1,t|t}, \widehat{X}_{2,t|t}  )\nonumber \\
 &+ I(X_{1,t}, X_{2,t};\widehat{X}_{1,t|t}, \widehat{X}_{2,t|t} |Y_1^{t-1}, Y_2^{t-1},\nonumber \\
 &\hst \widehat{X}_{1,t-1|t-1}, \widehat{X}_{2,t-1|t-1}  )\label{NRDF_M_nn_N_ad}\\
&\geq I(X_{1,t}, X_{2,t}; \widehat{X}_{1,t|t}, \widehat{X}_{2,t|t} |Y_1^{t-1}, Y_2^{t-1},\widehat{X}_{1,t|t-1}, \widehat{X}_{2,t|t-1}  ), \label{NRDF_M_nn_N_ae}
\end{align}
and equality holds in (\ref{NRDF_M_nn_N_ae}) if
\begin{align}
& \widehat{X}_{i,t|t} ={\bf E}\Big\{X_{i,t}\Big|Y_1^{t}, Y_2^{t}\Big\}=Y_{i,t}-a.s.,  \hso i=1,2.\label{NRDF_M_nn_N_af}
\end{align}
\noi (b) Consider the realization of part (a).  A sufficient condition for (\ref{NRDF_M_nn_N_af}) to hold is,  
\begin{align}
& {\bf E}\Big\{X_t\Big|Y_1^{t}, Y_2^t\Big\}
= {\bf E}\Big\{X_t\Big|Y_1^{t-1}, Y_2^{t-1}\Big\} \nonumber \\
&+\Cov\Big(X_t,Y_t|Y_1^{t-1}, Y_2^{t-1}\Big) \Cov\Big(Y_t,Y_t|Y_1^{t-1}, Y_2^{t-1}\Big)^{\dagger}\nonumber \\
&\Big(Y_t - {\bf E}\Big\{Y_t\Big|Y_1^{t-1}, Y_2^{t-1}\Big\}\Big)
= Y_t-a.s. \label{eq_st}
\end{align}
for $ t=1,\ldots$, provided such a $\big(H_t, Q_{(V_{1,t},V_{2,t})} \big)$ exists. \\
 Moreover,  if the pseudoinverse $\Cov\Big(Y_t,Y_t|Y_1^{t-1}, Y_2^{t-1}\Big)^{\dagger}=\Cov\Big(Y_t,Y_t|Y_1^{t-1}, Y_2^{t-1}\Big)^{-1}$ i.e.,  the inverse exists,  then the following,    Conditions 1 and 2,  are sufficient for (\ref{eq_st}) to hold {for $t=1, \ldots, n$}.
\begin{align}
& (1) \;\cov\Big(X_t, Y_t | Y_1^{t-1}, Y_2^{t-1}\Big)=\cov\Big(Y_t, Y_t | Y_1^{t-1}, Y_2^{t-1}\Big)  \label{cond_1}  \\
& (2) \;{\bf E}\left\{X_t \Big| Y_1^{t-1}, Y_2^{t-1}\right\}={\bf E}\left\{Y_t \Big| Y_1^{t-1}, Y_2^{t-1}\right\}. \label{cond_2}
\end{align}
\end{theorem}

\ \

\begin{proof} 
(a) The fact that a jointly Gaussian distribution is optimal,  is shown  similar to the classical RDF of Gaussian random processes with square error fidelity, and follows from \cite{stavrou-charalambous-charalambous-loyka2018siam}.   Hence,  the test channel distribution is conditionally Gaussian,  i.e.,  ${\bf P}_{Y_{1,t}, Y_{2,t}|Y_1^{t-1},Y_2^{t-1},X_{1,t},X_{2,t}}={\bf P}_{Y_{1,t}, Y_{2,t}|Y_1^{t-1},Y_2^{t-1},X_{1,t},X_{2,t}}^G$,  with linear conditional mean and nonrandom conditional covariance.   Such a distribution is induced by the parametric realization (\ref{real_1_an}) with linear $g_{i,t}(\cdot), i=1,2$.  (\ref{real_1_c}) follows from the  joint NRDF  given by  (\ref{NRDF_Markov_1_a}),  because for each $t$, the pay-off ${\bf E}\Big\{   \log \Big( \frac{{\bf P}_{Y_{1,t}, Y_{2,t}|Y_1^{t-1},Y_2^{t-1},X_{1,t},X_{2,t}}}{{\bf P }_{Y_{1,t},Y_{2,t} |Y_1^{t-1}, Y_2^{t-1}} }\Big)  \Big\}$ does not depend on $g_{i,t}(\cdot)$,  and   the average distortions $
{\bf E}\Big\{\sum_{t=1}^n ||X_{i,t}-Y_{i,t}||_{{\mathbb R}^{p_i}}^2\Big\}, i=1,2$ is minimized by $g_{i,t}(\cdot)$ given by (\ref{real_1_c}).   (i) The test channel distribution   (\ref{rep_11}) follows from the realization.   (ii) Equalities (\ref{NRDF_M_nn_N_ab})-(\ref{NRDF_M_nn_N_add}) follow from the realization and properties of conditional mutual information,  and the equality (\ref{NRDF_M_nn_N_ad})  follows from the chain rule of conditional mutual information \cite{pinsker:1964}.    Inequality (\ref{NRDF_M_nn_N_ae})  is due to the nonnegative property of conditional mutual information 
\begin{align}
 & I(X_{1,t}, X_{2,t}; Y_{1,t}, Y_{2,t} |Y_1^{t-1}, Y_2^{t-1},\nonumber \\
 & \hst \widehat{X}_{1,t-1|t-1}, \widehat{X}_{2,t-1|t-1}, \widehat{X}_{1,t|t}, \widehat{X}_{2,t|t}  )\geq 0.\label{mi_z}
\end{align}
 Moreover,  if  (\ref{NRDF_M_nn_N_af}) holds, then     the value of the  left hand side of (\ref{mi_z}) is zero, and the   inequality (\ref{NRDF_M_nn_N_ae}) holds with equality.  
  (b)    Since $(X_1^n, X_2^n, Y_1^n, Y_2^n)$ is jointly Gaussian, by mean-square estimation theory follows that if (\ref{eq_st}) holds then (\ref{NRDF_M_nn_N_af}) holds.      If the stated  inverse exists then (\ref{cond_1}),  (\ref{cond_2}) imply (\ref{eq_st}).
\end{proof}

Next,   we establish existence of the tuple $(H, Q_{(V_{1,t}, V_{2,t})})$ such that   equality holds in (\ref{NRDF_M_nn_N_af}), which is essential to   characterize the joint NRDF. \\

\begin{theorem} \label{thm:opt} 
Consider the joint NRDF $R_{X_1^n,X_2^n}(\Delta_1, \Delta_2)$ of (\ref{NRDF_A})   for the  tuple of  multivariate  Gaussian Markov process
of Definition~\ref{def:gmp}, with individual distortion criteria,  $d_{i,n}(x_1^n,y_1^n) \tri \frac{1}{n} \sum_{t=1}^n ||x_{1,t}-y_{{1},t}||_{{\mathbb R}^{p_i}}^2, i=1,2$. \\
The following hold.\\
(a) The  optimal test channel distribution of the joint NRDF is   conditionally Gaussian,  ${\bf P}_{Y_{1,t}, Y_{2,t}|Y_{1, {t-1}},Y_{2, {t-1}},X_{1,t},X_{2,t}}^G$,    induced by   $(X_1^n, X_2^n)$ and    the realization, for $t=1, \ldots, n$,  
 \begin{align}
\begin{pmatrix}
Y_{1,t} \\Y_{2,t}
\end{pmatrix}=&H_t \begin{pmatrix}
X_{1,t} \\X_{2,t}
\end{pmatrix}  +  \Big(I_{p_1+p_2}-H_t\Big) A_{t-1} \begin{pmatrix}Y_{1,t-1}\\ Y_{2, t-1} \end{pmatrix} \nonumber \\
&+\begin{pmatrix}
V_{1,t} \\V_{2,t}   \label{real_f}
\end{pmatrix}\\
 =&  H_tX_t+ \Big(I_{p_1+p_2}-H_t\Big)A_{t-1}Y_{t-1} + V_t
\end{align}
where the matrices, $(H_t,Q_{(V_{1,t},V_{2,t})})$ satisfy, 
\begin{align}
 &H_t \Sigma_{(E_{1,t}^-,E_{2,t}^-)} =\Sigma_{(E_{1,t}^-,E_{2,t}^-)}  - \Sigma_{(E_{1,t},E_{2,t})} \label{real_f1}\\
 &\hspace{2.0cm}=\Sigma_{(E_{1,t}^-,E_{2,t}^-)}  H_t\T \succeq 0,  \nonumber  \\
 &Q_{(V_{1,t},V_{2,t})}= H_t \Sigma_{(E_{1,t}^-,E_{2,t}^-)}  -H_t \Sigma_{(E_{1,t}^-,E_{2,t}^-)} H_t\T \succeq 0, 
\nonumber \\
&\Sigma_{(E_{1,t}^-, E_{2, t}^-)}=A_{t-1}\Sigma_{(E_{1,t-1}, E_{2, t-1})} A_{t-1}\T \nonumber \\
& \hspace{1.0cm} + B_{t-1} Q_{(W_{1,t}, W_{2,t})}B_{t-1}\T, \hso t=2, \ldots, n  \label{real_f3}\\
&\Sigma_{(E_{1,1}^-, E_{2, 1}^-)}=Q_{(X_{1,1}, X_{2,1})}.\label{real_f4}
\end{align}
If $Q_{(X_{1,1}, X_{2,1})}$ and  $B_{t-1} Q_{(W_{1,t}, W_{2,t})}B_{t-1}\T,  t=2, \ldots, n$ are full rank matrices then
\begin{align}
&Q_{(V_{1,t},V_{2,t})}=\Sigma_{(E_{1,t}E_{2,t})}\nonumber \\
&\hso  -\Sigma_{(E_{1,t}E_{2,t})} \Big(\Sigma_{(E_{1,t}^-,E_{2,t}^-)} \Big)^{-1}\Sigma_{(E_{1,t}E_{2,t})} \succeq 0. 
\end{align}

(b) The characterization of  the NRDF is given by 
\begin{align}
&R_{X_1^n,X_2^n}(\Delta_1, \Delta_2) =\inf_{{\cal M}_{X_1^n,X_2^n}^G(\Delta_1, \Delta_2)}\Big\{ \nonumber \\
& \hst   \sum_{t=1}^n I(X_{1,t}, X_{2,t}; Y_{1,t}, Y_{2,t}|Y_{1,t-1}, Y_{2, t-1})\Big\}   \\
 & =\inf_{{\cal M}_{X_1^n,X_2^n}^{G}(\Delta_1, \Delta_2)} \Big\{  \frac{1}{2} \sum_{t=1}^n  \log \Big(\frac{|\Sigma_{(E_{1,t}^-, E_{2, t}^-)}|}{|\Sigma_{(E_{1,t}, E_{2,t})}|}\Big)\Big\}\label{cost_f}
\end{align}
where the constraint set is 
\begin{align}
&{\cal M}_{X_1^n,X_2^n}^{G}(\Delta_1, \Delta_2) \tri \Big\{ \Sigma_{(E_{1,t}, E_{2,t})} \in {\cal S}_+^{p\times p}, \; t=1, \ldots, n \Big| \nonumber \\
& \Sigma_{(E_{1,t}, E_{2,t})}    \preceq \Sigma_{(E_{1,t}^-, E_{2,t}^-)},   t=1, \ldots, n,   \mbox{ (\ref{real_f3}),  (\ref{real_f4}) } \nonumber \\
 &   \frac{1}{n}\sum_{t=1}^n \trace(\Sigma_{E_{1,t}}) \leq \Delta_1, \;  \frac{1}{n}\sum_{t=1}^n \trace(\Sigma_{E_{2,t}}) \leq \Delta_2  \Big\}. 
\end{align} 
\end{theorem}

\ \

\begin{proof} First, note that by Theorem~\ref{thm:suffjoint}.(a),   the joint NRDF is also expressed as 
\begin{align}
&R_{X_1^n,X_2^n}(\Delta_1, \Delta_2) = \inf_{{\cal M}_{X_1^n,X_2^n}^G(\Delta_1, \Delta_2)}\Big\{ \\
&\sum_{t=1}^n I(X_{1,t}, X_{2,t}; Y_{1,t}, Y_{2,t}|Y_1^{t-1}, Y_2^{t-1})\Big\} \nonumber  \\
=\;&\inf_{ {\cal M}_{X_1^n,X_2^n}^G(\Delta_1,\Delta_2)}    
 {\bf E}\Big\{ \sum_{t=1}^n   \log \Big( \frac{{\bf P}_{X_{1,t}, X_{2,t}|Y_1^{t-1},Y_2^{t-1}}^G}{{\bf P }_{X_{1,t},X_{2,t} |Y_1^{t}, Y_2^{t}}^G }\Big)  \Big\} \label{NRDF_Markov_1_a_g} 
\end{align}
where ${\cal M}_{X_1^n,X_2^n}^G(\Delta_1, \Delta_2)$ is the subset of ${\cal M}_{X_1^n,X_2^n}(\Delta_1, \Delta_2)$ defined by (\ref{con}),  generated by jointly Gaussian distributions ${\bf P}_{X_1^t,X_2^t, Y_1^t, Y_2^t}^G$, $t=1, \ldots, n$ and ${\bf P}_{X_{1,t}, X_{2,t}|Y_1^{t-1},Y_2^{t-1}}^G$,  ${\bf P }_{X_{1,t},X_{2,t} |Y_1^{t}, Y_2^{t}}^G$ denote conditionally Gaussian  distributions,  obtained from the realization (\ref{real_1_an})-(\ref{real_1_ane}).  By properties of jointly Gaussian random processes,  then
\begin{align*}
&\Cov\Big(X_t,X_t\Big|Y^{t}\Big)={\bf E}\Big\{E_{t} \big(E_{t}\big)\T\Big\}=\Sigma_{(E_{1,t} E_{2,t})}\\
&\Cov\Big(X_t,X_t\Big|Y^{t-1}\Big)= {\bf E}\Big\{E_{t}^- \big(E_{t}^-\big)\T\Big\}=\Sigma_{(E_{1,t}^-, E_{2,t}^-)}
\end{align*}
Clearly,  $\Sigma_{(E_{1,t}^-, E_{2,t}^-)}$ is given by (\ref{real_f3}).      \\
(a) Realization   (\ref{real_f}) and specifically,  (\ref{real_f1})-(\ref{real_f4}) are obtained by    applying Theorem~\ref{thm:suffjoint}.(b)  so that (\ref{NRDF_M_nn_N_af}) holds. The conditions of  Theorem~\ref{thm:suffjoint}.(b) give rise to the equations of  $(H_t,Q_{(V_{1,t},V_{2,t})})$ as specified. 
(b) This  follows directly by using the realization of part  (a) to calculate (\ref{NRDF_Markov_1_a_g}).  
\end{proof}
 
\ \

The next theorem gives  the Kuhn-Tucker conditions of the optimization problem of  Theorem~\ref{thm:opt}.(b). \\

\begin{theorem}
\label{thm:char}
Consider $R_{X_1,X_2}(\Delta_1,\Delta_2)$ of Theorem~\ref{thm:opt}.(b), defined by  (\ref{cost_f}) and assume   $\overline{Q}_{t}\tri B_{t}Q _{(W_{1,t+1}, W_{2,t}+1)}B_{t}\T\succ 0$, and  $R_{X_1^n,X_2^n}(\Delta_1,\Delta_2)< +\infty$. The Lagrange functional is, 
\begin{align}
&{\cal L}\tri  \sum_{t=1}^n \bigg\{   \frac{1}{2}\log\Big( \frac{| \Sigma_{(E_{1,t}^-,E_{2,t}^-)}|}{|\Sigma_{(E_{1,t},E_{2,t}}\big)| } \Big)\nonumber \\
& + \trace \Big (\Theta_t \Big( \Sigma_{(E_{1,t},E_{2,t})} - \Sigma_{(E_{1,t}^-,E_{2,t}^-)} \Big)\Big )   -\trace \Big (V_t\Sigma_{(E_{1,t},E_{2,t})} \Big )\bigg\} \nonumber \\
&+ \lambda_{1} \Big (\sum_{t=1}^n \trace \Big (\Sigma_{E_{1,t}}\Big )  - n \Delta_{1} \Big ) + \lambda_2 \Big (\sum_{t=1}^n \trace \Big (\Sigma_{E_{2,t}}\Big )  - n\Delta_{2} \Big )   \nonumber
\end{align}
where $ \Theta_t \succeq 0$, $V_t \succeq 0$, $\lambda_i \in [0,\infty), i=1,2$. \\
The optimal $\{\Sigma_{(E_{1,t}, E_{2,t})}: t=1, \ldots, n\}  \in {\cal M}_{X_1^n,X_2^n}^{G}(\Delta_1, \Delta_2)$ for $ R_{X_1^n,X_2^n}(\Delta_{1},\Delta_{2})$ is found as follows.  \\
(i) Stationarity:
\begin{align}
-\frac{1}{2} \Sigma_{(E_{1,t},E_{2,t})}^{-1}  +  \begin{bmatrix}
\lambda_1 I_{p_1} & 0 \\ 0 & \lambda_2 I_{p_2}
\end{bmatrix} + \Theta_t +V_t=0 . \label{eq:ErrorCovLag}
\end{align}
(ii) Complementary Slackness:
\begin{align}
&\lambda_1 \Big (\sum_{t=1}^n \trace \Big (\Sigma_{E_{1,t}}\Big )  - n \Delta_{1} \Big ) = 0,  \\
& \lambda_2 \Big (\sum_{t=1}^n \trace \Big (\Sigma_{E_{2,t}}\Big )  -n \Delta_{2} \Big ) = 0,  \label{eq:CompSlackLambda} \\
&\trace \Big (V_t\Sigma_{(E_{1,t},E_{2,t})}\Big ) = 0, \hso t=1, \ldots, n, \\
& \trace \Big (\Theta_t \Big( \Sigma_{(E_{1,t},E_{2,t})} - \Sigma_{(E_{1,t}^-,E_{2,t}^-)} \Big)\Big ) = 0.  \label{eq:CompSlackTheta}
\end{align}
(iii) Primal Feasibility: Defined  by ${\cal M}_{X_1^n,X_2^n}^{G}(\Delta_1, \Delta_2$.\\
(iv) Dual Feasibility: $\lambda_1 \geq 0, \hso \lambda_2 \geq 0, \hso \Theta_t \succeq 0, \hso V_t \succeq 0, t=1, \ldots, n$.\\ Moreover, the following hold.\\
(a)  $V_t=0$ for $t=1, \ldots, n$ and   
\begin{align}
&\mbox{For $t=n$}: \nonumber \\
 &\Sigma_{(E_{1,n},E_{2,n})}  = \frac{1}{2}\Bigg ( \begin{bmatrix}
\lambda_1 I_{p_1} & 0 \\ 0 & \lambda_2 I_{p_2}
\end{bmatrix} + \Theta_n \Bigg)^{-1} \succ 0 . \label{eq:ErrorCovLag_new}\\
&\mbox{For $t=n-1, \ldots, 1$}: \nonumber \\
&\Sigma_{(E_{1,t},E_{2,t})}+ \Sigma_{(E_{1,t},E_{2,t})}\overline{Q}_t \Sigma_{(E_{1,t},E_{2,t})}\nonumber \\
& -\frac{1}{2}\Bigg ( \begin{bmatrix}
\lambda_1 I_{p_1} & 0 \\ 0 & \lambda_2 I_{p_2}
\end{bmatrix} + \Theta_t- A_t\T \Theta_{t+1}A_t \Bigg)^{-1}=0 \label{eq:ErrorCovLag_new_a}
\end{align}
(b) If  $\Sigma_{(E_{1,t}^-,E_{2,t}^-)} - \Sigma_{(E_{1,t},E_{2,t})}  \succ 0$ for all $t=1, \ldots, n$  
then $\Theta_t=0$ for all $t=1, \ldots, n$,  and

\end{theorem} 
\begin{proof} The derivation is utilizes \cite[Theorem~5.3]{stavrou-charalambous-charalambous-loyka2018siam}.
\end{proof}

\ \

\begin{example} {To illustrate fundamental challenges, we consider $X_{i,t}: \Omega \rightarrow \mathbb{R},\; i=1,2$ and $Y_{i,t}: \Omega \rightarrow \mathbb{R},\; i=1,2$ for $t=1,\dots,n$. For simplicity, assume  $\overline{Q}_t = \diag(q_1,q_2) \succ 0$, i.e., $q_1>0,\;q_2>0$, and \small $A_t = A = \begin{bmatrix}
a_{11} & a_{12}\\ a_{21} & a_{22}
\end{bmatrix}t=0,\dots,n$, and $X_{1,1}, X_{2,1}$ are independent with variances $\sigma_{1}^2, \sigma_2^2$, respectively. \normalsize By Theorem \ref{thm:char} item (b),} {we have:
For $t=n$:}

\begin{align}
\Sigma_{(E_{1,n},E_{2,n})}  = \begin{bmatrix}
 \frac{1}{2\lambda_1} & 0 \\ 0 & \frac{1}{2\lambda_2}
\end{bmatrix}, 
& \Sigma_{E_{1,n},E_2,n} = 0.
\end{align}
{For  $t=1, \ldots, n-1$:}
\begin{align}
&\Sigma_{(E_{1,t},E_{2,t})}+ \Sigma_{(E_{1,t},E_{2,t})}\overline{Q}_t \Sigma_{(E_{1,t},E_{2,t})} -\begin{bmatrix}
 \frac{1}{2\lambda_1} & 0 \\ 0 & \frac{1}{2\lambda_2}
\end{bmatrix}=0 \nonumber
\end{align}
{The set of equations for $t=1, \ldots, n-1$,  are }
\begin{align}
& \Sigma_{E_1,t}+ \Sigma_{E_1,t}^2 q_1 + {\Sigma_{E_{1,t},E_2,t}^2} q_2 - \frac{1}{2\lambda_1} = 0\label{eq:a}\\
&{\Sigma_{E_{1,t},E_2,t}}(1+\Sigma_{E_1,t} q_1 + \Sigma_{E_2,t} q_2)  = 0 \label{eq:diag}\\
&\Sigma_{E_2,t}+ \Sigma_{E_2,t}^2 q_2 + {\Sigma_{E_{1,t},E_2,t}^2} q_1 - \frac{1}{2\lambda_2} = 0\label{eq:b}
\end{align}
By \eqref{eq:diag}, $\Sigma_{E_1,E_2,t} = 0$ or $\Sigma_{E_1,t} q_1 + \Sigma_{E_2,t} q_2 = -1$. The later cannot hold because the left hand side  is always positive, hence  {$\Sigma_{E_{1,t},E_2,t} = 0$} and the optimal matrices $\Sigma_{(E_{1,t},E_{2,t})}$ are diagonal for all $t=1,\dots,n-1$. Then by \eqref{eq:a} and \eqref{eq:b} the positive solutions are
\begin{align}
\Sigma_{E_1,t} = \frac{-1+\sqrt{1+\frac{2q_1}{\lambda_1}}}{2q_1},\;\Sigma_{E_2,t} = \frac{-1+\sqrt{1+\frac{2q_2}{\lambda_2}}}{2q_2}\nonumber
\end{align}
{For $t=0,\ldots, n$ using the above values we can determine  $\lambda_i \ge 0,\; i=1,2$, using the average distortions as follows.}
\begin{align}
\frac{1}{2\lambda_i} + (n-1)\bigg (\frac{-1+\sqrt{1+\frac{2q_i}{\lambda_i}}}{2q_i}\bigg) = n\Delta_i,\hso i=1,2
\end{align}
Then by using \eqref{real_f3}, {need to determine the matrix $\Sigma_{(E_{1,t}^-, E_{2, t}^-)}$. Suppose its  diagonal entries  are $\alpha_t, \beta_t$ and its non-diagonal entry is  $\gamma_t$.  Then the equations for $\Sigma_{(E_{1,t}^-, E_{2, t}^-)}$,  for $t=2, \ldots, n$, are  given by 
\begin{align}
&\alpha_t = a_{11}^2\Sigma_{E_1,t-1} + a_{12}^2\Sigma_{E_2,t-1} + q_1, \label{1}\\
&\beta_t = a_{21}^2\Sigma_{E_1,t-1} + a_{22}^2\Sigma_{E_2,t-1} + q_2, \\
&\gamma_t = a_{11}a_{21}\Sigma_{E_1,t-1} + a_{22}a_{12}\Sigma_{E_2,t-1} \label{3}
\end{align}
For $t=1, \alpha_1=\sigma_1^2, \beta_1=\sigma_2^2, \gamma_1=0$.  }
Therefore, the joint NRDF is given by,
{
\begin{align}
&R_{X_1^n,X_2^n}(\Delta_1, \Delta_2) = \frac{1}{2} \sum_{t=1}^n  \log \Big(\frac{\alpha_t\beta_t-\gamma_t^2}{\Sigma_{E_1,t}\Sigma_{E_2,t}}\Big) \label{eq:exrdf}\\
& =  \frac{1}{2}  \log \Big(\frac{\alpha_1\beta_1-\gamma_1^2}{\Sigma_{E_1,1}\Sigma_{E_2,1}}\Big)+ \frac{(n-2)}{2} \log \Big(\frac{\alpha_t\beta_t-\gamma_t^2}{\Sigma_{E_1,t}\Sigma_{E_2,t}}\Big)\nonumber \\
&\hso+\frac{1}{2}  \log \Big(\frac{\alpha_n\beta_n-\gamma_n^2}{\Sigma_{E_1,n}\Sigma_{E_2,n}}\Big)
\end{align}
and \eqref{eq:exrdf} holds for $(\Delta_1,\Delta_2)$ such that  $\Sigma_{(E_{1,t}^-,E_{2,t}^-)} - \Sigma_{(E_{1,t},E_{2,t})}  \succ 0$ for all $t=1, \ldots, n$. The closed form calculations of this region as a function of $(\lambda_1,\lambda_2,A,\overline{Q}_t)$ are lengthy hence omitted due to space limitation.  
The per unit time limit is, then obtain from the solution of  (\ref{1})-(\ref{3}), at any $t \in \{2, \ldots, n-1\}$, i.e., 
\begin{align}
&\lim_{n \longrightarrow \infty} \frac{1}{n}R_{X_1^n,X_2^n}(\Delta_1, \Delta_2) = \frac{1}{2}\log \Big(\frac{\alpha_t\beta_t-\gamma_t^2}{\Sigma_{E_1,t}\Sigma_{E_2,t}}\Big).
\end{align}
}
\end{example}

\ \ 
 
\begin{remark} For the sub-set of the rate region for which Theorem~\ref{thm:char}.(b) holds, it is possible to compute the closed-form expression of the joint NRDF $R_{X_1^n,X_2^n}(\Delta_1, \Delta_2)$, by using (\ref{eq:ErrorCovLag_new}), (\ref{eq:ErrorCovLag_new_a}), with $\Theta_t=0, t=1, \ldots, n$. This will lead to a generalization of   the classical joint RDF given in  \cite[Theorem~III.3]{stylianou2021joint},  for  the tuple of jointly independent and identically (IID) distributed Gaussian process $(X_1^n, X_2^n)$, i.e., ${\bf P}_{X_{1,t}, X_{2,t}}={\bf P}_{X_1, X_2}, t=1, \ldots, n$, with individual distortion criteria   \cite[Theorem~III.3]{stylianou2021joint}. The case $\Sigma_{(E_{1,t},E_{2,t})} - \Sigma_{(E_{1,t}^-,E_{2,t}^-)}  \succeq 0$ but not $\Sigma_{(E_{1,t},E_{2,t})} - \Sigma_{(E_{1,t}^-,E_{2,t}^-)}  \succ 0$ is to our experience, a challenging problem, even for the simplest application example of a tuple of IID process. 
\end{remark}

%
%
%
%
\section{Conclusions and open problems}\label{sec:conclusions}
The joint nonanticpative RDF is analyzed  for a tuple of random process with  individual fidelity criteria.  Achievable lower bound and structural properties of the test channels are derived. The application example of a tuple of jointly multivariate Gaussian Markov process with two square-error fidelity criteria is analyzed.  A  fundamental open problem which is not addressed in this paper  is  the computation of the joint nonanticipative RDF of Theorem~\ref{thm:opt}. Although, for a tuple of  multivariate Gaussian Markov process  this is  challenging problem (i.e., currently the only known solution is,   for a tuple of  scalar, IID Gaussian random variables  \cite{xiao}), some progress is expected.

\bibliographystyle{ieeetr}
\bibliography{references}

%
%
%
%
\end{document}